\documentclass[11pt,letterpaper]{article}

\usepackage[margin=1in]{geometry}
\usepackage{microtype}
\usepackage{graphicx}
\usepackage{caption}
\usepackage{subcaption}
\usepackage{booktabs} % for professional tables
\usepackage[normalem]{ulem}

\usepackage{tikz} 
\usepackage{pgfplots}
\pgfplotsset{compat=1.14}
\usepgfplotslibrary{fillbetween}
\usepackage{hyphenat}
\usepackage{subcaption}

\usepackage{graphicx}
\graphicspath{{figures/}}
\def\input@path{{figures/}}

\renewcommand{\P}[1]{{\mathbb{P}}\left(#1\right)}
\usepackage{enumerate}
\usepackage{amssymb}
\usepackage{amsmath}
\usepackage{nicefrac}
\usepackage{graphicx}
\usepackage{tikz,pgfplots}
\usepackage[hidelinks,breaklinks,bookmarks=false]{hyperref}

\usepackage{algorithm}
\usepackage[noend]{algorithmic}

\usepackage{bbm}
\usepackage[amsmath,amsthm,hyperref,thmmarks]{ntheorem}
\usepackage{adjustbox}
\usepackage{multirow}
\usepackage{xspace}

\usepackage{thm-restate}

\usepackage{natbib}
\usepackage{array}
\newtheorem{theorem}{Theorem}[section]
\newtheorem{corollary}[theorem]{Corollary}
\newtheorem{lemma}[theorem]{Lemma}

\newtheorem{definition}[theorem]{Definition}
\usepackage[capitalize, nameinlink]{cleveref}
\crefname{theorem}{Theorem}{Theorems}
\Crefname{lemma}{Lemma}{Lemmas}
\Crefname{claim}{Claim}{Claims}
\Crefname{fact}{Fact}{Facts}
\Crefname{observation}{Observation}{Observations}
\Crefname{invariant}{Invariant}{Invariants}

\DeclareMathOperator{\OPT}{OPT}

\newcommand{\eps}{\varepsilon}

\newcommand{\bbRp}{{\mathbb{R}_{\ge 0}}}
\newcommand{\rb}[1]{\left( #1 \right)} % round Brackets
\newcommand{\marginal}[2]{f\rb{#1 \mid #2}}
\newcommand{\cM}{\mathcal{M}}
\newcommand{\BBAlg}{\textsc{Alg}\xspace}

\DeclareMathOperator*{\argmin}{argmin}
\DeclareMathOperator*{\argmax}{argmax}
\DeclareMathOperator{\polylog}{polylog}

\newcommand{\Sopt}{\OPT}
\newcommand{\I}{\mathcal{I}}

\newcommand{\ind}[1]{\mathbb{I}_{\left\{#1\right\}}}

\newcommand{\cl}{\text{cl}}

\newcommand{\cS}{\mathcal{S}}
\newcommand{\cE}{\mathcal{E}}

\newcommand{\cW}{\mathcal{W}}

\renewcommand{\P}[1]{{\mathbb{P}}\left(#1\right)}
\newcommand{\E}[1]{{\mathbb{E}}\left[#1\right]}

\title{Deletion Robust Non-Monotone Submodular Maximization over Matroids
}

\author{
	Paul D{\"u}tting\thanks{Google Research. Email:   
\texttt{\{\href{mailto:duetting@google.com}{duetting}, \href{mailto:ashkannorouzi@google.com}{ashkannorouzi},
		\href{mailto:silviol@google.com}{silviol}, \href{mailto:zadim@google.com}{zadim}\}@google.com}}
	\and
	Federico Fusco\thanks{Department of Computer, Control, and Management 
Engineering ``Antonio Ruberti'', Sapienza University of Rome, 
Italy. Email:   
\texttt{\href{mailto:fuscof@diag.uniroma1.it}{fuscof}@diag.uniroma1.it}.
Part of this work was done while Federico was an intern at Google Research, hosted by Paul D{\"u}tting.}
\and
Silvio Lattanzi{$^*$}
\and
Ashkan Norouzi-Fard{$^*$}
\and
Morteza Zadimoghaddam{$^*$}}

\begin{document}

\maketitle

\begin{abstract}
Maximizing a submodular function is a fundamental task in machine learning and in this paper we study the deletion robust version of the problem under the classic matroids constraint. Here the goal is to extract a small size summary of the dataset that contains a high value independent set even after an adversary deleted some elements. We present constant-factor approximation algorithms, whose space complexity depends on the rank $k$ of the matroid and the number $d$ of deleted elements. In the centralized setting we present a $(4.597+O(\eps))$-approximation algorithm with summary size $O( \frac{k+d}{\eps^2}\log \frac{k}{\eps})$ that is improved to a $(3.582+O(\eps))$-approximation with $O(k + \frac{d}{\eps^2}\log \frac{k}{\eps})$ summary size when the objective is monotone.  In the streaming setting we provide a $(9.435 + O(\eps))$-approximation algorithm with summary size and memory $O(k + \frac{d}{\eps^2}\log \frac{k}{\eps})$; the approximation factor is then improved to  $(5.582+O(\eps))$ in the monotone case. 
\end{abstract}

\section{Introduction}

Submodular maximization is a fundamental problem in machine learning that encompasses a broad range of applications, including active learning~\citep{GolovinK11} sparse reconstruction~\citep{Bach10,DasK11,DasDK12}, video analysis~\citep{ZhengJCP14}, and data summarization~\citep{lin-bilmes-2011-class,BairiIRB15}.

Given a submodular function $f$, a universe of elements $V$, and a family $\mathcal{F} \subseteq 2^V$ of feasible subsets of $V$, the optimization problem consists in finding a set $S \in \mathcal{F}$ that maximizes $f(S)$. A natural choice for $\mathcal{F}$ are capacity constraints (a.k.a.~$k$-uniform matroid constraints) where any subset $S$ of $V$ of size at most $k$ is feasible. Another standard restriction, which generalizes capacity constraints and naturally comes up in a variety of settings, are matroid constraints. As an example where 
%the more general matroid constraints 
such more general constraints are needed, consider a movie recommendation application, where, given a large corpus of movies from various genres, we want to come up with a set of recommended videos that contains at most one movie from each genre.

Exact submodular maximization is a NP-hard problem, but efficient 
% natural 
algorithms exist that obtain small constant-factor approximation guarantees in both the centralized and in the streaming setting \citep[e.g.,][]{fisher78-II,CalinescuCPV11,BuchbinderF19,ChakrabartiK15,FeldmanK018}.

In this work we design algorithms for submodular optimization over matroids that are robust to deletions. 
A main motivation for considering deletions are privacy and user preferences. 
For example, users may exert their ``right to be forgotten'' or may update their preferences and thus exclude some of the data points. For instance, in the earlier movie recommendation example, a user may mark some of the recommended videos as ``seen'' or ``inappropriate,'' and we may wish 
% to be able 
to quickly update the list of recommendations. 

\subsection{The Deletion Robust Approach}
Following \citet{MitrovicBNTC17},
%\citet{MirzasoleimanK017}\fnote{maybe we want to cite \cite{MitrovicBNTC17} instead}, 
we model robustness to deletion as a two phases game against an adversary. 
% We adopt the following approach for modeling robustness to deletions. There are two phases: i
In the first phase, the algorithm receives a robustness parameter $d$ and chooses a subset $W \subseteq V$ as summary of the whole dataset $V$. Concurrently, an adversary selects a subset $D\subseteq V$ with $|D| \le d$. The adversary may know the algorithm but has no access to its random bits. In the second phase, the adversary reveals $D$ and the algorithm determines a feasible solution from $W \setminus D$. The goal of the algorithm is to be competitive with the optimal solution on $V \setminus D$. Natural performance metrics in this model are the algorithm's approximation guarantee 
% $\alpha$ 
and its space complexity as measured by the size of the set $W$. We consider this problem in both the centralized and in the streaming setting. % \textcolor{red}{
The adversary studied is called oblivious, as its choice of the deleted elements $D$ \textit{does not} depend on the realized $W$, but possibly on the structure of the algorithm (not on its random bits). This is a natural assumption in the applications. Consider the movie recommendation example: the fact that a movie has already been watched by the user or that it is deemed inappropriate is independent of the fact that the specific movie has been selected or not in the summary.
% }

%We assume that the adversary can choose the elements to be deleted knowing the algorithm we run, but that it is oblivious to the random choices of the algorithm. 

In this model, to obtain a constant-factor approximation, the summary size has to be $\Omega(k+d)$, even when $f$ is additive and the constraint is a $k$-uniform matroid. To see this, consider the case where exactly $k+d$ of the elements have unitary weight and the remaining elements have weight zero. The adversary selects $d$ of the valuable elements to be deleted, but the algorithm does not know which. To be protective against any possible choice of the adversary, the best strategy of the algorithm is to choose $W$ uniformly at random from the elements that carry weight. This leads to an expected weight of the surviving elements of $|W|\cdot k/(k+d)$, while the optimum is $k$.

%Even without deletions, the set $W$ kept by the algorithm has to be of size $\Omega(k)$ to be able to provide a constant-factor approximation. The simple reason is that when all elements have the same weight $w$, with a summary $W$ of size $o(k)$ we can only obtain weight $o(kw)$ while the optimum is of order $\Omega(kw)$.
%With $d$ deletions, a similar argument shows that the space complexity has to be of order $\Omega(k+d)$. \textcolor{red}{PD: Explain why?}

Prior work gave deletion robust algorithms for the special case of $k$-uniform matroids and monotone objective. The state-of-the-art   is a $2+O(\eps)$ approximation with $O(k + \frac{d \log k}{\eps^2})$ space in the centralized setting, while in the streaming setting the same approximation is achievable at the cost of an extra multiplicative factor of $O(\frac{\log k}{\eps})$ in space complexity \citep{KazemiZK18}. 
%In a closely related work~\cite{MirzasoleimanK017} the authors enlarge the set of constraints for which deletion robust algorithms exist; enabling their application in the broad range of submodular optimization problems beyond capacity constraints. For deletion robust submodular optimization under matroids constraint they present a $4$-approximation algorithm that use memory $O(kd\log\abs{V}^{O(1)})$. The main open question from their work is to design a \emph{space-efficient} algorithm for the problem. This question is important in many practical scenario where datasets are large and space is an important resource.

For general matroids, \citet{MirzasoleimanK017} give a black-box reduction that adapts (non-robust) submodular optimization algorithms to the deletion robust setting. Their approach yields a $3.147$-approximation algorithm for monotone objectives using \citet{Ashkan21} and a $5.205$ approximation for non-monotone objectives via \cite{BuchbinderF19}. The downside of \citet{MirzasoleimanK017} is that the space complexity is $\tilde O(kd)$.\footnote{Where the $\tilde{O}$ notation hides logarithmic factors.} 
The multiplicative factor $d$ is inherent in the construction, and the $k$ is needed in any subroutine, so this approach necessarily leads to a space complexity of $\Omega(dk)$. The main open problem from their work is to design \emph{space-efficient} algorithms for the respective problems. This question is important in many practical scenarios where datasets are large and space is an important resource.

\begin{table}[!t]
\begin{tabular}{lllll}
\hline
Reference & Objective    & Model       & Approximation                & Summary size                    \\[1.5pt]
\hline
\Cref{cor:centralized_nonmonotone} (this work)  & Non-Monotone & Centralized & $4.597 + O(\eps)$              & $O\big(\frac{k+d}{\eps^2}\log k\big)$   \\
\Cref{cor:centralized_monotone} (this work)  & Monotone     & Centralized & $3.582 + O(\eps)$ & $O\big(k + \frac{d}{\eps^2}\log k\big)$ \\[1.5pt] \hline
\cite{MirzasoleimanK017}    & Non-Monotone & Streaming   & $5.205$                             & $\tilde{O}(kd)$                         \\
\Cref{thm:robust-streaming} (this work)  & Non-Monotone & Streaming   & $9.435 + O(\eps)$            & $O\big(k + \frac{d}{\eps^2}\log k\big)$ \\
\cite{MirzasoleimanK017}    & Monotone     & Streaming   & $3.147$                      & $\tilde{O}(kd)$                         \\
\Cref{cor:streaming} (this work) & Monotone     & Streaming   & $5.582+O(\eps)$                       & $O\big(k + \frac{d}{\eps^2}\log k\big)$\\
\citet{Zhang22}     & Monotone     & Streaming   &  $4$                          & $O(k + d)$                              \\[1.5pt]
\hline
\end{tabular}
\caption{The table summarizes the results for deletion robust submodular maximization with matroid constraints. Note that the $\tilde O$ hides terms poly-logarithmic in the rank $k$. The approximation guarantees of \cite{MirzasoleimanK017} are obtained via the state-of-the-art streaming algorithms for submodular maximization with matroid constraints (\cite{BuchbinderF19} and \cite{Ashkan21}).}
\label{table:our_results}
\end{table}

\subsection{Our Results}

We present the first constant-factor approximation algorithms for deletion robust submodular maximization subject to general matroid constraints with almost optimal space usage, i.e., our algorithms only use $\tilde{O}(k+d)$ space. 
More formally, in the centralized setting we present a $(4.597 + O(\eps))$-approximation algorithm with summary size $O(\frac{k+d}{\eps^2}\log \frac{k}{\eps})$, that can be improved to a $(3.582+O(\eps))$-approximation with summary size $O(k + \frac{d}{\eps^2}\log \frac{k}{\eps})$ if the objective is monotone. In the streaming setting we provide a $(9.435+O(\eps))$-approximation algorithm with summary size and memory $O(k + \frac{d}{\eps^2}\log \frac{k}{\eps})$, that becomes a $(5.582+O(\eps))$-approximation if the objective is monotone. We remark that, besides the black box results by \cite{MirzasoleimanK017} that have a far-from-optimal space usage, we are the first to offer a positive result for non-monotone submodular functions in the deletion robust setting, with {\em any} type of constraint.
All our results are summarized and compared to the existing literature in \Cref{table:our_results}.

For monotone objectives the constants in the two cases are $2+\beta$ and $4+\beta$, where $\beta$ is $e/(e-1) \approx 1.582$, i.e., the best-possible approximation guarantee for the standard centralized problem \citep{CalinescuCPV11,Feige98}. For the non-monotone case we still have a $2+\beta$, where $\beta \approx 2.597$ is the state-of-the-art approximation guarantee for non-robust centralized submodular maximization subject to matroid constraint \citep{BuchbinderF19}; our approximation for the streaming setting also depends on the routine used, but in a more intricate way. The result of \cite{BuchbinderF19} are not known to be tight, thus better algorithms may exist and would automatically improve our approximation guarantees. 

We point out that the state-of-the-art approximation guarantees for (non-robust) submodular maximization with matroid constraint in the streaming model are $3.147$-approximation in the monotone case and a $5.205$-approximation in the general non-negative, i.e. possibly non-monotone, case \citep{Ashkan21}. The ``price of robustness'' is thus just an extra small additive term in the approximation factor, depending on the setting. At the same time, up to possibly a logarithmic factor, the memory requirements of all our results are tight. 

\subsection{Our Techniques}

Intuitively, the extra difficulty in obtaining space-efficient robust deletion summaries with general matroid constraints is that the algorithm can only use elements respecting the matroid constraint to replace a good deleted element from a candidate solution. This issue gets amplified when multiple elements need to be replaced. This is in sharp contrast to the specal case of $k$-uniform matroids where all elements can replace any other element.

Our algorithms start by setting a logarithmic number of value thresholds that span the average contribution of relevant optimum elements, and use these to group together elements with similar marginal value. The candidate solution is constructed using only elements from bundles that are large enough (at least a factor of $1/\eps$ larger than the number of deletions). Random selection from a large bundle protects/insures the value of the selected solution against adversarial deletions. 
In the centralized algorithm, it is possible to sweep through the thresholds in decreasing order. This monotonic iteration helps us design a charging mechanism for high value optimum elements dismissed due to the matroid constraint. We use the matroid structural properties to find an injective mapping from optimum elements rejected by the matroid property to the set of selected  elements. Given the monotonic sweeping of thresholds, the marginal value of missed opportunities cannot dominate the values of added elements. Finally, since each fixed element of the base set has very low probability to be added to the solution (due to the deletion robust sampling), we can handle {\em for free} the non-monotonicity of the objective in the analysis via a well known sampling argument for submodular functions \citep{FeigeMV07}.

In the streaming setting, elements arrive in an arbitrary order in terms of their membership to various bundles. So we keep adding elements as long as a large enough bundle exists. Addition of elements from lower value bundles might technically prevent us from selecting some high value elements due to the matroid constraint. So when considering a new element, we allow for a swap operation with any of the elements in the solution to maintain feasibility of the matroid constraint. We perform the swap if the marginal value of the new element $e$ is substantially (a constant factor) higher than the marginal value that the element $e'$ we are kicking out of the solution had when it was added to the solution. The constant factor gap between marginal values helps us account for not only $e'$ but also the whole potential chain of elements that $e'$ caused directly or indirectly to be removed in the course of the algorithm. As we repeatedely swap elements in and out the solution, it no longer holds that each fixed elements of the base set has very low probability to be, at some point, added to the solution. To tackle non-monotonicity, we embed a uniform sampling step in the procedure outlined above \citep[as in, e.g.,][]{FeldmanK018,AmanatidisFLLMR21,AmanatidisFLLR22}: any new element is discarded with some fixed probability before even being considered for swapping.

%\subsection{Experiments}

% Paragraph about empirical results

% By testing our algorithms on multiple real-world datasets, we validate that they achieve almost optimal values while storing only a small fraction of the elements. In the settings we tried they typically attain at least $90-95\%$ of the value output by state-of-the-art algorithms that know the deletions in advance even though we only keep a few percents of the elements. Our algorithms persevere in achieving high value solutions and maintaining a concise memory footprint even in the face of large number of deletions. 

% \textcolor{red}{
Finally, a consideration on the deletion robust parameter $d$. Our algorithms (just like the algorithms from earlier work) do not need exact knowledge of $d$; an upper bound on it yields the same approximation, at the cost of a larger summary. Also, with some extra work, it is possible to show that our approximation guarantees degrade gracefully when the estimate of $d$ was too conservative and there were actually more deletions. 
Note that some prior knowledge on $d$ is a reasonable assumption in the applications and that without any prior information about $d$ it is provably impossible to achieve good approximation and memory efficiency.
% }
\subsection{Related Work}

Robust submodular optimization has been studied for more than a decade. In \citet{krause2008robust}, the authors study robustness from the perspective of multiple agents each with its own submodular valuation function; their objective is to select a subset that maximizes the minimum among all the agents' valuation functions. In some sense, this maximum minimum objective could be seen as a max-min fair subset selection goal.   Another robustness setting that deals with multiple valuation function is distributionally robust submodular optimization \citep{staib2019distributionally} in which we have access to samples from a distribution of valuations functions. These settings are fundamentally different from the robustness setting we study in our paper. 

\citet{orlin2018robust} and \citet{bogunovic2017robust} looked at robustness of a selected set in the presence of a few deletions. In their model, the algorithm needs to finalize the solution before the adversarial deletions are revealed and the adversary sees the choices of the algorithm. Therefore having at least $k$ deletions reduces the value of any solution to zero. Here, $k$ is the cardinality constraint or the rank of the matroid depending on the setting. This is the most prohibitive deletion robust setting we are aware of in the literature and not surprisingly the positive results of \citet{orlin2018robust} and \citet{bogunovic2017robust} are mostly useful when we are dealing with a few number of deletions. 

%dynamic setting
\citet{MirzasoleimanK017} study submodular maximization, and provide a general framework to empower insertion-only algorithms to process deletions on the fly as well as insertions. 
Their result works on general constraints including matroids. 
As a result, they provide dynamic algorithms that process a stream of deletions and insertions with an extra multiplicative overhead of $d$ (on both the computation time and memory footprint) compared to insertion-only algorithms. Here $d$ is the overall number of deletions over the course of the algorithm. They propose the elegant idea of running $d+1$ concurrent streaming submodular maximization algorithms where $d$ is the maximum number of deletions. 
Every element is sent to the first algorithm. If it is not selected, it is sent to the second algorithm; if it is not selected again, it is sent to the third algorithm and so on. With this trick, they maintain the invariant that the solution of one of these algorithms is untouched by the adversarial deletions, and therefore this set of $\tilde{O}(dk)$ elements suffice to achieve robust algorithms for matroid constraints. This approach has the drawback of having per update computation time linearly dependent on $d$ which can be prohibitive for large number of deletions. It also has a total memory of $\tilde{O}(dk)$ which could be suboptimal compared to the lower bound of $\Omega(d+k)$. 
For a closely related dynamic setting with cardinality constraint, \citet{LattanziMNTZ20} provide a $2$-approximation algorithm 
for $k$-uniform matroids with per update computation time of poly-logarithmic in the total number of updates (length of the stream of updates). This much faster computation time comes at the cost of a potentially much larger memory footprint (up to the whole ground set).
\citet{monemizadeh2020dynamic} independently designed an algorithm with similar approximation guarantee and an update time quadratic in the cardinality constraint with a smaller dependence on logarithmic terms.

As argued in the Introduction, a simple lower bound of $\Omega(d+k)$ exists on the memory footprint needed by any robust constant-factor approximation algorithm. The gap between this lower bound and the $O(dk)$ memory footprint in \citet{MirzasoleimanK017} motivated the follow-up works that focused on designing even more memory and computationally efficient algorithms. 

% two phases deletion robust
\citet{MitrovicBNTC17} and  \citet{KazemiZK18} are the first to study the deletion robust setting we consider in our work. They independently
designed submodular maximization algorithms for $k$-uniform matroids. \citet{MitrovicBNTC17} proposed a streaming algorithm that achieves constant competitive ratio with memory $O((k + d) \polylog(k))$. Their results extends to the case that the adversary is aware of the summary set $W$ before selecting the set $D$ of deleted elements. On the other hand, \citet{KazemiZK18} design centralized and distributed algorithms with constant factor approximation and $O(k + d\log(k))$ memory, as well as a streaming algorithm with $2$-approximation factor and memory footprint of $O(k \log(k) + d\log^2(k))$. 
We have borrowed some of their ideas including bundling elements based on their marginal values and ensuring that prior to deletions, elements are added to the solution only if they are selected uniformly at random from a large enough bundle (pool of candidates). 

Subsequently, \citet{AvdiukhinMYZ19} showed how to obtain algorithms for the case of knapsack constraints with similar memory requirements and constant factor approximation. While their approximation guarantees are far from optimal, their notion of robustness is stronger than the one we consider: there the adversary can select the set of deleted elements adaptively with respect to the summary produced by the algorithm.  

We aim to achieve the generality of the work of  \citet{MirzasoleimanK017} work by providing streaming algorithms that work for all types of matroid constraints while maintaining the almost optimal computation time and space efficiency of \citet{MitrovicBNTC17} and \citet{KazemiZK18}. 

% sliding window
Another well-studied robustness model is the sliding window model. In this case, the deletions occur as the algorithm sweeps through the stream of elements and in that sense they occur regularly rather than in an adversarial manner. Every element is deleted exactly $W$ steps after its arrival. So at every moment, the most recent $W$ elements are present and the objective is to select a subset of these present elements. \citet{epasto2017submodular} design a $3$ approximation algorithm for the case of cardinality constraints with a memory independent of the window size $W$. \citet{zhao2019submodular} provide algorithms that extend the sliding window algorithm to settings where elements have non-uniform lifespans and leave after arbitrary times.

%Non-robust literature
Prior to deletion robust models and motivated by large scale applications, \citet{MirzasoleimanKSK13} designed the distributed greedy algorithm and showed that it achieves provable guarantees for the cardinality constraint problem under some assumptions. In follow-up work, \citet{mirrokni2015randomized} provided core-set frameworks that always achieve constant factor approximation in the distributed setting. \citet{barbosa2016new} showed how to approach the optimal $e/(e-1)$ approximation guarantee by increasing the round complexity of the distributed algorithm. For the case of matroid constraints, \citet{ene2019submodular} provided distributed algorithms that achieve the $e/(e-1)$ approximation with poly-logarithmic number of distributed rounds. For the streaming setting, \citet{Ashkan21} provided a $3.147$ competitive ratio algorithm with $\Tilde{O}(k)$ memory with a matroid constraint of rank $k$. 

A preliminary version of this work considering only monotone objectives has appeared as \cite{DuettingFLNZ22}. We mention that, in follow-up work, \citet{Zhang22} design a simple algorithm for the streaming version of the deletion-robust submodular maximization problem with monotone objective subject to a $p$-matroid constraint. Their algorithm has optimal space usage of $O(k + d)$ and achieves a $4p$ approximation. They handle robustness via an elegant non-uniform sampling technique that avoids the extra $O(\log k)$ due to keeping the different buckets.

\section{Problem formulation and preliminary results}
We consider a set function $f: 2^V \to \bbRp$ on a (potentially large) ground set $V$. Given two sets $X, Y \subseteq V$, the \emph{marginal gain} of $X$ with respect to $Y$ quantifies the change in value of adding $X$ to $Y$ and is defined as
\[
	\marginal{X}{Y} = f(X \cup Y) - f(Y) \,.
\]
When $X$ consists of a singleton $x$, we use the shorthand $f(x|Y)$ instead of $f(\{x\}|Y)$. The function $f$ is called \emph{monotone} if $\marginal{e}{X}  \geq 0$ for each set $X \subseteq V$ and element $e \in V$, and \emph{submodular} if for any two sets $X$ and $Y$ such that $X \subseteq Y \subseteq V$ and any element $e \in V \setminus Y$ we have $ \marginal{e}{X} \ge \marginal{e}{Y}.$
Throughout the paper, we assume that $f$ is given in terms of a value oracle that computes $f(S)$ for given $S \subseteq V$ and that $f$ is \emph{normalized}, i.e., $f(\emptyset) = 0$. When not specified otherwise, we do not assume monotonicity of the function $f$. We slightly abuse the notation and for a set $X$ and an element $e$, use $X+e$ to denote $X \cup \{e\}$ and $X - e$ for $X \setminus \{e\}$. 

A non-empty family of sets $\cM \subseteq 2^V$ is called a \emph{matroid} if it satisfies the following properties. \textit{Downward-closedness}: if $A \subseteq B$ and $B \in \cM$, then $A \in \cM$; \textit{augmentation}: if $A, B \in \cM$ with $|A| < |B|$, then there exists $e \in B$ such that $A + e \in \cM$. We call a set $A \subseteq 2^V$ \emph{independent}, if $A \in \cM$, and \emph{dependent} otherwise. An independent set that is maximal with respect to inclusion is called a {\em base}; all the bases of a matroid share the same cardinality $k$, which is referred to as the {\em rank} of the matroid. In general, it is possible to define the rank set function $r: 2^V \to \mathbb{N}$ that outputs the cardinality of the largest independent set: $r(T) = \max\{|I| \text{ s.t. } I \subseteq T \text { and } T\in \cM\}$. Dually to basis, for any $A\notin \cM$, we can define a {\em circuit} $C(A)$ as a minimal dependent subset of $A$, i.e. $C(A) \notin \cM$ such that all its proper subsets are independent. Starting from any matroid $\cM$ it is possible to define two auxiliary structures:  the \emph{contraction} of $\cM$ by $S$, written $\cM / S$, which is the matroid on $V\setminus S$ with rank function $r(T) = r_\cM(T \cup S) - r_\cM(S)$, and the \emph{restriction} of $\cM$ by $S$, which is the matroid $\cM'$ on $E \setminus S$ whose independent sets $T$ are characterized by the property that $T \cup S \in \cM$.

\subsection{The Deletion Robust Model}
The deletion robust model consists of two phases. The input of the first phase is the ground set $V$ and a robustness parameter $d$, while the input of the second phase is an adversarial set of $d$ deleted elements $D \subset V$, along with the outputs of the first phase. The goal is to design an algorithm that constructs a small size summary $W \subseteq V$ that is robust to deletions in the first phase, and a solution $S \subseteq W \setminus D$ that is independent with respect to matroid $\cM$ in the second phase. The difficulty of the problem lies in the fact that the summary $W$ has to be robust against {\em any} possible choice of set $D$ by an adversary oblivious to the randomness of the algorithm. 

For any set of deleted elements $D$, the optimum solution denoted by $\OPT(V \setminus D)$ is defined as
\[
f(\OPT(V \setminus D)) = \argmax_{R \subseteq V \setminus D, R \in \cM} f(R).
\]
We say that a two phase algorithm is an $\alpha$ approximation for this problem if for all $ D \subset V \text{ s.t. } |D| \le d, $ it holds that 
%\[
%    \max_{D \subset V, |D| \le d}\frac{\OPT(V \setminus D)}{\E{f(S_D)}} \le \alpha
%\]
\[
    f(\OPT(V \setminus D)) \le \alpha \cdot \E{f(S_D)}, 
\]
where $S_D$ is the solution produced in Phase II when set $D$ is deleted and the expectation is with respect to the eventual internal randomization of the algorithm. Besides $\alpha$, an important feature of a two phase algorithm is its summary size, i.e., the cardinality of the set $W$ returned by the first phase. 

In this paper we also consider the streaming version of the problem where the elements in $V$ are presented in some arbitrary order (Phase I) and at the end of such online phase the algorithm has to output a deletion robust summary $W.$ Finally, the deleted set $D$ is revealed and Phase II on $W \setminus D$ takes place offline. The quality of an algorithm for the streaming problem is not only assessed by its approximation guarantee and summary size, but is also measured in terms of its {\em memory}, i.e., the cardinality of the buffer in the online phase.

\subsection{Preliminary Results}

We conclude this section with three preliminary results: the folklore sampling Lemma by \cite{BuchbinderFNS14}, a combinatorial Lemma we derived from the well known Hall's Marriage Theorem and a graph theoretical result from \cite{FeldmanK018}. While the first Lemma is our main tool to address non-monotonicity, the other two results are used for the charging arguments in the analysis of our algorithms.

We start with the sampling Lemma, that has been proved by \citet{FeigeMV11} and that we report as in Lemma 2.2 of \citet{BuchbinderFNS14}.

\begin{lemma}[Sampling Lemma]
	\label{lem:sampling}
	Let $f: 2^V \to \mathbb{R}_{\ge 0}$ be a (possibly not normalized) submodular set function, let $ X \subseteq 
	V$ and let $X(p)$ be a sampled subset, where each element of $X$ appears 
	with probability at most $p$ (not necessarily independent). Then $\mathbb{E} \left[ f(X(p)) \right] \geq (1-p)f(\emptyset)\,$.
\end{lemma}

We move now to our combinatorial Lemma. At a high level, it is a general tool that maps injectively elements in a dependent set to those of an independent set under some conditions. 

\begin{lemma}[Combinatorial Lemma]
\label{lem:hall-lemma}
Consider a matroid $\cM \subseteq 2^V$ and two sets $F \subseteq V$, and $G \in \mathcal{M}$. Suppose that for all $x \in G\setminus F$ there exist a set $F_x \subseteq F, F_x \in \mathcal{M}$ such that $F_x  + x \not\in \mathcal{M}$. Then there exists a mapping $h: G \setminus F \rightarrow F$ such that 
\begin{itemize}
    \item for all $x \in G\setminus F$, $h(x) \in F_x$, and
    \item for all $x,y \in G\setminus F$ with $x \neq y$, $h(x) \neq h(y)$.
\end{itemize}
\end{lemma}

Intuitively, $h$ as a {\emph{semi-matching}} that matches all elements in $G \setminus F$ to an element in $\cup_{x \in G\setminus F} F_x$, while some of the elements in $\cup_{x \in G\setminus F} F_x \subseteq F$ can remain unmatched. To prove the Lemma, we use the combinatorial version of Hall's Marriage Theorem, that concerns set systems and the existence of a transversal (a.k.a.~system of distinct representatives). 
Formally, let $\cS$ be a family of finite subsets of a base set $X$ ($\cS$ may contain the same set multiple times). A \emph{transversal} is an injective function $f: \cS \rightarrow X$ such that $f(S) \in S$ for every set $S \in \cS$. In other words, $f$ selects one representative from each set in $S$ in such a way that no two of these representatives are equal.  

\begin{theorem}[Hall (1935)] \label{hall-thm}
A family of sets $\cS$ has a transversal if and only if $\cS$ satisfies the marriage condition: i.e., if for each subfamily $\cW \subseteq \cS$, it holds that $|\cW| \leq \left|\bigcup_{F \in \cW} F\right|.$
\end{theorem}

We are ready to present the proof of our combinatorial lemma.
\begin{proof}
    Let $W$ be any subset of $G\setminus F$ and define $F_W = \bigcup_{x \in W} F_x$. We want to show that 
    \[
    |F_W| = |\bigcup_{x \in W} F_x| \geq |W|. 
    \]

    In order to do that we first show that $\cl(W \cup F_W) = \cl(F_W)$, where $\cl(.)$ denotes the \emph{closure} (or \emph{span}) of a set. We have that $W \subseteq \cl(F_W)$, in fact, for each element $x \in W$ there exists a subset $F_x \subseteq F_W$ such that $F_x + x \not\in \I$ and therefore $x \in \cl(F_x)$, which implies, by monotonicity of the closure with respect to the inclusion, that $x \in \cl(F_W)$.
    We also know that $F_W \subseteq \cl(F_W)$, hence $F_W \cup W \subseteq \cl(F_W)$. If we apply the closure to both sets, we get
    \[
        \cl(F_W \cup W) \subseteq \cl(\cl(F_W)) = \cl(F_W) \subseteq \cl(F_W\cup W),
    \]
    where the equation follows by the well-known properties of closure. This shows that $\cl(F_W \cup W) = \cl(F_W)$ as claimed.
    
    Now let's look at the restriction of the matroid $\cM$ to $\cl(F_W \cup W) = \cl(F_W)$. Afterwards, contract this matroid by $(F \cap G) \cap \cl(F_W)$ and call the resulting matroid with $\cM'$. We claim that $W$ is independent in this new matroid $\cM'$. This is due to, $W \subseteq G \setminus F$ and $W \cup (F \cap G) \subseteq G$ and $G \in \I$. If we call $r'$ the rank of the matroid $\cM'$, we have $
    r_\cM(\cl(F_W)) \geq r_\cM(\cl(F_W) \setminus (F \cap B) ) = r' \geq |W|.$ Finally, it holds that $
    |F_W| \geq r_\cM(F_W) = r_\cM(\cl(F_W)). $
    Putting these two chains of inequalities together we obtain $|F_W| \geq |W|$ as claimed. The proof follows by applying \Cref{hall-thm}.
\end{proof}

The last preliminary result we present is Lemma $13$ of \cite{FeldmanK018}. It concerns directed acyclic graphs (DAGs) where the nodes are also elements of a matroid. Under some assumption, it guarantees the existence of an injective mapping between elements in an independent set and of the sinks of the DAG. As a convention, we denote with $\delta^+(u)$ the out-neighborhood of any node $u.$  

\begin{lemma}
\label{lem:graph}
Consider an arbitrary directed acyclic graph $G = (V, E)$ whose vertices are elements of some matroid $\cM$. If every non-sink vertex $u$ of $G$ is spanned by $\delta^+(u)$ in $\cM$, then for every set $S$ of vertices of $G$ which is independent in $\cM$ there must exist an injective function $\psi$ such that, for every vertex $u \in S$, $\psi(u)$ is a sink of $G$ which is reachable from $u$.
\end{lemma}

\section{Centralized Algorithm} \label{section:cent}
In this section, we present a centralized algorithm for the Deletion Robust Submodular Maximization problem subject to a matroid constraint. To that end, we start by defining some notations. Let $\tilde e$ be the $(d+1)$-th element with largest value according to $f$ and let $\Delta$ be its value. Given the precision parameter $\eps$, we define the set of relevant thresholds $T$ as follows 
\[
    T = \left\{ (1 + \eps)^i\mid \eps \cdot \frac{\Delta}{(1+\eps)k} < (1 + \eps)^i \le \Delta \right\}.
\]

The first phase of our algorithm constructs the summary in iterations using two main sets: a candidate solution $A$ and a reservoir $B$ of good elements, that are grouped into buckets of similar elements. The algorithm goes over the thresholds in $T$ in decreasing order and updates sets $A, B$, and $V$. Let $\tau$ be the threshold considered at some point, then $B_\tau$ contains any element $e \in V$ such that $f(e | A) > \tau$ and $ A+ e \in \cM$. These are the high contribution elements that can be added to $A$. As long as the size of $B_\tau> (k+d)/\eps$, the algorithm chooses uniformly at random an element from $B_\tau$, adds it to $A$ and recomputes $B_\tau$. We observe that $A$ is robust to deletions, i.e., when $d$ elements are deleted, the probability of one specific element in $A$ being deleted is intuitively at most $\eps$. Moreover, since each element added to $A$ is drawn from a pool of elements with {\em similar} marginals, the value of this set after the deletions decreases at most by a factor $(1-\eps)$ in expectation (this is a very slack bound). 

As soon as the cardinality of $B_{\tau}$ drops below $(k+d)/\varepsilon$, no more elements can be added directly from it to $A$ while keeping $A$ robust and feasible. Therefore, we remove these elements from $V$ and save them for Phase II. During the execution of the algorithm we need to take special care of the top $d$ element with highest $f$ values. To avoid complications, we remove them from the instance before starting the procedure and add them to set $B$ at the end. This does not affect the general logic and only simplifies the presentation and the proofs. The pseudocode of the centralized algorithm for Phase I is given in \Cref{alg:centralized-phase-I}.

\begin{algorithm}[t]
\caption{Centralized Algorithm Phase I} \label{alg:centralized-phase-I}
\begin{algorithmic}[1]
\STATE \textbf{Input:} Precision $\varepsilon$ and deletion parameter $d$
\STATE $\Delta \gets $ $(d+1)^{th}$ largest value in $\{f(e) \mid e \in V\}$.
\STATE $V_d \gets $ elements with the $d$ largest values.
\STATE $V \gets V\setminus V_d$, $A \gets \emptyset$.
\STATE $T = \left\{ (1 + \eps)^i\mid \eps \cdot \frac{\Delta}{(1+\eps)k} < (1 + \eps)^i \le \Delta \right\}$
\FOR{$\tau \in T$ in decreasing order}
    \STATE $B_{\tau} \gets \{e \in V  \mid A + e \in \cM, f(e \mid A) \geq \tau \}$ %\COMMENT{The index $t$ varies in $T$}
    \WHILE{$|B_\tau| \geq \frac{k+d}{\varepsilon}$}
    \STATE $e \gets$ a random element sampled independently and uniformly from $B_{\tau}$.
    \STATE $V \gets V - e$, $A \gets A + e$.
    \STATE $B_{\tau} \gets \{e \in V  \mid A + e \in \cM, f(e \mid A) \geq \tau \}$
    \ENDWHILE
    \STATE Remove $B_{\tau}$ from $V$.
\ENDFOR
\STATE $B \gets V_d \cup \bigcup_{\tau \in T} B_{\tau}$.
\RETURN $A, B$.
\end{algorithmic}
\end{algorithm}

The summary $W$ computed at the end of Phase I is composed by the union of $A$ and $B$ and is then passed to the second phase of our algorithm, which uses as routine an arbitrary algorithm \BBAlg for submodular maximization subject to matroid constraint. \BBAlg takes as input a set of elements, function $f$ and matroid $\cM$ and returns a $\beta$-approximate solution. In this phase we simply use \BBAlg to compute a solution among all the elements in $A$ and $B$ that survived the deletion and return the best among the computed solution and the value of the surviving elements in $A$.

\begin{algorithm}[t]
\caption{Algorithm Phase II} \label{alg:phase-II}
\begin{algorithmic}[1]
\STATE \textbf{Input:} $A$ and $B$ outputs of phase I, set $D$ of deleted elements and optimization routine $\BBAlg$
\STATE $A' \gets A \setminus D$, $B' \gets B \setminus D$
\STATE $\tilde S \gets \BBAlg(A' \cup B')$
\RETURN $S \gets \argmax \{f(A'), f(\tilde S)\}$
\end{algorithmic}
\end{algorithm}

\begin{theorem}
\label{thm:robust-centralized}
    Consider the problem of deletion robust non-monotone submodular maximization with matroid constraints.
    For $\eps \in (0, 1/5)$, the Centralized Algorithm (\Cref{alg:centralized-phase-I} and \Cref{alg:phase-II}) is in expectation a $(2+ \beta + O(\eps))$-approximation algorithm with summary size $O(\frac{k + d}{\eps^2}\log{\frac k\eps})$, where $\beta$ is the approximation ratio of the auxiliary algorithm \BBAlg.
\end{theorem}
\begin{proof}
In this proof, sets $A$ and $B$ are the sets returned by Algorithm~\ref{alg:centralized-phase-I} at the end of phase I.
We start by analyzing the size of the summary, composed by $A$ and $B$. Set $A\in \cM$, thus its size is no more than the rank of $\cM$, therefore: $|A| \leq k$. Set $B$ is the union of $V_d$ (containing exactly $d$ elements) and $\bigcup_{\tau \in T} B_{\tau}$. Each set $B_{\tau}$ has at most $\frac{k+d}{\eps}$ element and there are at most $\frac{2}{\eps}\log \frac{k}{\eps}$ such sets. Therefore
\[
    |A| + |B| \leq  k + d + \frac{2(k+d)}{\eps^2}\log\frac{k}{\eps} \in O\left(\frac{k + d}{\eps^2}\log{\frac k\eps}\right).
\]

We move our focus on bounding the approximation factor. To that end, we fix any set $D$ with $|D| \le d$ and bound the ratio between the expected value of $f(S)$ and $f(\OPT)$ (we omit the dependence on $D$ since it is clear from the context). As a first step, we bound the value of $\OPT$ with that of $\OPT \cup A$. This passage is trivially true when $f$ is monotone, but needs some work otherwise. 

\begin{lemma}
\label{lem:nonmonotone}
    It holds that $(1-\eps) f(\OPT) \le  \E{f(\OPT \cup A)}.$
\end{lemma}
\begin{proof}
    Every time a new random element $x_i$ is added to the candidate solution $A$, it is drawn uniformly at random from a large pool of elements of cardinality at least $(k+d)/\eps$. Thus $\P{x = x_i} \le \frac{\eps}{k+d}$ for all fixed $x$ and all $i = 1, \dots k.$ By union bound, we have the following:
    \[
        \P{x \in A} \le \sum_{i=1}^k\P{x = x_i} \le \eps \frac{k}{k+d} \le \eps, \; \;
        \forall x \in V.
    \]
    Consider the submodular function $g: 2^V \to \mathbb{R}_{\ge 0}$, defined as $g(T) = f(\OPT \cup T),$ for all $T \subseteq V$. We can apply the sampling Lemma (\Cref{lem:sampling}) on $g$ and the bound on $\P{x \in A}$ to obtain the desired claim:
    \[
        \E{f(\OPT \cup A)} = \E{g(A)} \ge (1-\eps) g(\emptyset) = (1-\eps) f(\OPT).
    \]
\end{proof}

This means we just need to bound the value of $\OPT \cup A$; to do that we introduce three subsets and analyze the resulting terms separately. First, recall that  we denoted by $B'$ the subset of $B$ surviving the deletion: $B' = B \setminus D$; clearly it holds that $B \cap \OPT \subseteq B'$. Then, define $L$ as the set of elements of low value:
\[
    L = \left\{e \in V \setminus D \mid  f(e\mid A) \le \eps \cdot \frac{ f(\OPT)}{k}\right\}.
\]
Intuitively, $L$ contains all the low value elements: these elements do not increase the value of the submodular function considerably if added to $A$ (we do not lose much by ignoring them). Finally, let $F$ be the surviving elements that have not been added to the summary because of feasibility issues; formally, $F = V \setminus (A \cup B \cup L \cup D)$. We can now decompose the contributions of elements in $\OPT \cup A$ using the sets we introduced: by submodularity and \Cref{lem:nonmonotone}, we in fact have:
\begin{align}
    \nonumber
    (1-\eps)f(\OPT) \le & \E{f(\OPT \cup A)} \\
    \nonumber
    \le & \E{f(A)} + \E{f(\OPT \cap B' \mid A)} \\
    \label{eq:partition_OPT}
    + &  \E{f(\OPT \cap F \mid A)}  
    + \E{f(\OPT \cap L \mid A)}
\end{align}

Recall that $A'$ is the candidate solution once the deleted elements are realized; we use its expected value to bound the first term of the inequality. Intuitively, for any element that is part of $A$ the probability of it being deleted is $O(\eps)$, since it is sampled from $(d+k)/\eps$ {\em similar} elements uniformly at random and only $d$ elements are deleted. We formalize this idea in the following Lemma.

\begin{lemma}
\label{lem:robustness}
    It holds that $\E{f(A)} \le \frac{1+\eps}{1-\eps}\ \E{f(A')}$.
\end{lemma}
\begin{proof}
    The proof is similar to that of Lemma 2 in \citet{kazemi17arxiv}, but note that we do assume the monotonicity of $f$. For the sake of the analysis, for each threshold $\tau \in T$, let $A_{\tau}$ denote the set of elements in $A$ that were added during an iteration of the for loop corresponding to threshold $\tau$ in \Cref{alg:centralized-phase-I}. Moreover, let $n_{\tau} = |A_{\tau}|$ and order the elements in $A_{\tau} = \langle x_{1,\tau},x_{2,\tau},\dots, x_{n_{\tau},\tau}\rangle$ according to the order in which they are added to $A_{\tau}.$ Finally, let $I(\ell,\tau)$ be the indicator random variable corresponding to the element $x_{\ell,\tau}$ not being in $D$. Note that the randomness here is with respect to the random draw of the algorithm in Phase I: $D$ is fixed but unknown to the algorithm. The crucial argument is that $x_{\ell,\tau}$ is drawn uniformly at random from a set of cardinality at least $(d+k)/\varepsilon$ where at most $d$ elements lie in $D$, hence 
    \begin{equation}
    \label{eq:robust_proba}
        \P{I(\ell,\tau)} \ge 1- \varepsilon\frac{d}{k+d} \ge 1 - \eps.
    \end{equation}
    We can decompose the value of $A$ as follows using the definition of marginal value:
        \begin{align*}
        f(A) &= \sum_{\tau \in T} \sum_{\ell=1}^{n_{\tau}} f(x_{\ell,\tau}|\cup_{t > \tau}A_t \cup \{  x_{1,\tau}, x_{2,\tau}, \dots, x_{\ell-1,\tau}\}).
    \end{align*}
    The elements added in a specific iteration of the for loop share the same marginal up to an $(1+\varepsilon)$ factor, i.e.,
    \begin{equation} \label{eq:at_cont}
        \tau \le f(x_{\ell,\tau}|\cup_{t>\tau}A_t \cup \{  x_{1,\tau}, x_{2,\tau}, \dots, x_{\ell-1,\tau}\}) \le \tau \cdot (1+\varepsilon), \quad \forall \tau \in T, \quad \forall \ell = 1, \dots, n_{\ell}.        
    \end{equation}
    Thus, summing up those contributions, we have
    \begin{equation}
    \label{eq:robust_ineq}
       \sum_{\tau \in T} |A_{\tau}| \cdot \tau \le f(A) \le (1+\varepsilon) \sum_{\tau \in T} |A_{\tau}| \cdot \tau.  
    \end{equation}
    
    We now decompose in a similar way the value of $A' = A \setminus D$. 
    Let $I(\ell,i) \cdot x$ be a shorthand to denote the element $x$ if the indicator variable $I(\ell,i)$ is $1$ and the empty set otherwise.
    Similarly, we let $I(\ell,i) \cdot X$ be a shorthand to denote the set $X$ if the indicator variable $I(\ell,i)$ is $1$ and the empty set otherwise. We also define $A'_{\tau} = A_{\tau} \setminus D = A_{\tau} \cap A'$, we have
    \begin{align*}
        f(A') &= \sum_{\tau \in T} \sum_{\ell=1}^{n_{\tau}} I(\ell,\tau) \cdot f(x_{\ell,\tau}|\cup_{t > \tau}A_t \cup \{I(1,\tau) \cdot x_{1,\tau}, I(2,\tau) \cdot x_{2,\tau}, \dots, I(\ell-1,\tau) \cdot x_{\ell-1,\tau}\})\\
        &\ge \sum_{\tau \in T} \sum_{\ell=1}^{n_\tau} I(\ell,\tau)\cdot f(x_{\ell,\tau}|\cup_{t > \tau}A_t \cup \{x_{1,\tau}, x_{2,\tau}, \dots, x_{\ell-1,\tau}\})\\
        &\ge \sum_{\tau \in T} \sum_{\ell=1}^{n_{\tau}} I(\ell,\tau) \cdot \tau\\
        &=\sum_{\tau \in T} |A'_{\tau}| \cdot \tau,
    \end{align*}
    where the first inequality follows by submodularity and the second one follows from \Cref{eq:at_cont}. We apply the expected value to the previous inequality, which results in
    \[
        \E{f(A')} \ge \sum_{\tau \in T} \E{|A'_{\tau}|} \cdot \tau \ge (1-\varepsilon)\sum_{\tau \in T} \E{|A_{\tau}|} \cdot \tau \ge \frac{1-\varepsilon}{1
        + \varepsilon}\E{f(A)},
    \]
    where the second inequality follows by linearity of expectation and \Cref{eq:robust_proba}, while the last one from the right hand side of  \Cref{eq:robust_ineq}. 
\end{proof}

For the second term of \Cref{eq:partition_OPT}, observe that $\OPT \cap B' \in \cM$ and is contained in the set of elements passed to \BBAlg, so its value is dominated by $\beta$ times the value of the of $\BBAlg(A' \cup B')$, all in all:
\begin{align} 
    \E{f(\OPT \cap B' \mid A)} \le
    %\beta \cdot \E{f(\tilde S)} \le 
    \beta \cdot \E{f(S)}. \label{eq:main-offline-2}
\end{align}
Recall infact that $S$ is the final solution output by the algorithm. Bounding the third term is more involved and is handled in the following Lemma via a charging argument that uses our Combinatorial Lemma (\Cref{lem:hall-lemma}).

\begin{lemma} \label{lemma:offline-third-eq}
    It holds that $\E{f(\OPT \cap F \mid A)} \leq (1+\eps) \E{f(A)}$.
\end{lemma}
\begin{proof}
    Let's order the elements in $\Sopt \cap F = \{x_1, x_2, \dots \}$ according to the order in which they were removed from some $B_{\tau}$ because of feasibility constraint (since they are in $V\setminus (A \cup B \cup D \cup L)$ it must be the case). Note that once an element fails the feasibility test and is removed from some $B_{\tau}$, it never becomes feasible again because of diminishing returns property of submodular functions. Furthermore, for each such $x_i$ let $F_i$ be the set $A$ when $x_i$ fails the feasibility test. We know that $F_i$ is independent since $A$ is independent during the execution of Algorithm~\ref{alg:centralized-phase-I}; moreover, by definition, $F_i + x_i \not \in \cM $. By \Cref{lem:hall-lemma} there exists an injective function $h:\Sopt \cap F \to A$ such that $h(x_i) \in F_{i}$ for all $i$. Let $A_i$ denote the set $A$ when the element $h(x_i)$ gets added to it. Since $h(x_i)$ has been added to $A$ before $x_i$, it holds that $A_i \subseteq F_i$. We have
    \begin{align*}
        f(\Sopt \cap F \mid A) &\le  \sum_{x \in \Sopt \cap F} f(x|A_i) \\
        &\le  \sum_{x \in \Sopt \cap F} (1+\varepsilon) \cdot f(h(x_i)|A_i) \\
        &\le  (1+\varepsilon) \cdot f(A),
    \end{align*}
    where the first inequality follows from submodularity and the fact that $A_i \subseteq A$ for all $i$. The second inequality uses the fact that if $x_i$ was still feasible to add when $h(x_i)$ was added, then their marginals are at most a $(1+\varepsilon)$ factor away. The last inequality follows from a telescopic decomposition of $A$, the fact that an element is added to the candidate solution $A$ only if its marginal is positive (actually larger than a certain non-negative threshold), and the fact that $h$ is injective. Applying the expectation to both extremes of the chain of inequalities gives the Lemma.
\end{proof}

The fourth term of \Cref{eq:partition_OPT} refers to at most $k$ elements and can be bounded based on the definition of $L$: any element $e$ in $L$ is such that $f(e|A) \leq \eps \cdot \frac{f(\OPT)}{k},$ by submodularity we have then that
\begin{align}
\label{eq:main-offline-4}
% \[
    \E{f(\OPT \cap L \mid A)} \le \E{\sum_{e \in \OPT \cap L} f(e\mid A)} \le \eps \cdot f(\OPT).
% \]
\end{align}

Finally, we compose the four bounds in \Cref{eq:partition_OPT}, thus obtaining the following:
    \begin{align*}
        (1-\eps)f(\OPT) \le &\E{f(A)} + \E{f(\OPT \cap B' \mid A)} \\
        &+ \E{f(\OPT \cap F \mid A)}  
    + \E{f(\OPT \cap L \mid A)} \\
    \le & \left(2 + \eps \right)\E{f(A)} + \beta \E{f(S)} + \eps f(\OPT) \tag*{(Due to (\ref{eq:main-offline-2}), (\ref{eq:main-offline-4}) and Lemma \ref{lemma:offline-third-eq})} \\
    \le & (2 + \eps)\frac{1 + \eps}{1-\eps}\E{f(A')} + \beta \cdot \E{f(S)} + \eps f(\OPT) \tag*{(Lemma \ref{lem:robustness})}\\
    \le & \left[(2 + \eps)\frac{1 + \eps}{1-\eps} + \beta\right] \cdot \E{f(S)} + \eps \cdot f(\OPT) \tag*{(Definition of $S$)}
    \end{align*}

    Rearranging terms we get
    \begin{align*}
        f(\OPT) &\le \frac{\E{f(\OPT \cup A)}}{1-2\eps}\le \left[ \frac{\left(2 + \eps \right)(1+\eps)}{(1-2\eps)(1-\eps)} + \frac{\beta}{1-2\eps}\right] \cdot \E{f(S)} \\
        &\le \left[ 2 + \beta +  \left(4 \beta + 15 \right) \cdot \eps \right] \cdot \E{f(S)},
    \end{align*}
    where in the last inequality we used that $\frac{\left(2 + \eps \right)(1+\eps)}{(1-2\eps)(1-\eps)} \le 2 + 15 \eps$ for all $\eps \in (0,\frac 15)$ and that $1/(1-2\eps) \le (1+4\eps)$ in the same interval. The theorem follows since as long as $\beta$ is a constant it holds that $\left(4 \beta + 15\right) \cdot \eps \in O(\eps)$.
\end{proof}

Plugging in the previous Theorem some routines for non-monotone submodular maximization with matroid constraints we get, for examples, an approximation of $\approx 4.597$ if we use the state of the art algorithm as in \cite{BuchbinderF19} or $\approx 4.73$ if we use the less involved measured continuous greedy \citep{FeldmanNS11}. Using the former routine and handling the $\eps$ term accordingly, it is easy to see we have the following Corollary. 

\begin{corollary}
\label{cor:centralized_nonmonotone}
    Consider the problem of deletion robust non-monotone submodular maximization with matroid constraints and fix any constant $\delta \in (0,1)$, then there exists a constant $C_{\delta}$ such that for any $\eps \in (0,\delta)$, in expectation a ${(4.597+\eps \cdot C_{\delta})}$-approximation algorithm with summary size $O(\frac{(k + d)}{\eps^2}\log \frac{k}{\eps})$ exists.
\end{corollary}

If we move our attention to monotone objective, we get an approximation factor $4$ if we use as subroutine the greedy algorithm \citep{fisher78-II} or $2 + \tfrac{e}{e-1} \leq 3.582$ if we use continuous greedy as in \citet{CalinescuCPV11}. Furthermore, it is possible to slightly improve the bound on the cardinalily of the summary.

\begin{corollary}
\label{cor:centralized_monotone}
    Consider the problem of deletion robust monotone submodular maximization with matroid constraints and fix any constant $\delta \in (0,1)$, then there exists a constant $\tilde C_{\delta}$ such that for any $\eps \in (0,\delta)$, in expectation a ${(3.582+\eps \cdot \tilde C_{\delta})}$-approximation algorithm with summary size $O(k + \frac{d}{\eps^2}\log \frac{k}{\eps})$ exists.
\end{corollary}

\begin{proof}
    First, we study the improved bound on the summary size. Since the function is monotone, we do not need to argue as in \Cref{lem:nonmonotone} to bound $f(\OPT)$ with $f(\OPT \cup A)$; thus we can modify the condition on the while loop of Algorithm \ref{alg:centralized-phase-I} to simply consider $|B_{\tau}| \ge d/\eps$. With this tweak, there are still at most $\frac 1{\eps} \log(k/\eps)$ buckets, but each one has cardinality at most $d/\eps$, thus the summary size is $O(k +  \frac{d}{\eps^2}\log \frac{k}{\eps}).$
    From the approximation point of view, \Cref{lem:robustness} still holds, as the crucial property that $\P{I(\ell,\tau)} \ge 1-\eps$ is verified (we are sampling uniformly at random from a bucket of cardinality at least $d/\eps$ where there are at most $d$ deleted elements). The rest of the analysis of \Cref{thm:robust-centralized} goes trough without any other difference and we get that 
    \[
        f(\OPT) \le \E{f(\OPT \cup A)} \le \left[\left(2 + \eps \right)\frac{1+\eps}{1-\eps} + \beta\right] \cdot \E{f(S)} + \eps \cdot f(\OPT)       
    \]
    We use as optimization routine \BBAlg continuous greedy, therefore we can plug in $\beta = \frac{e}{e-1}$ and, rearranging the terms we obtain
    \begin{align*}
        f(\Sopt) &\le \left(\frac{e}{(e-1)(1-\eps)} + \frac{(2+\eps)(1+\eps)}{(1-\eps)^2}\right) \E{f(S)}\\
        &\le \Big{[}\frac{e}{e-1} + 2  + \underbrace{\frac{8 e - 7 -\delta(2e-1)}{(e - 1) (\delta - 1)^2}}_{\tilde C_{\delta}} \cdot \eps \Big{]}\E{f(S)},
    \end{align*}
    where the last inequality can be numerically verified and holds for any  $\eps \in (0,\delta).$
\end{proof}

\section{Streaming Setting} \label{section:streaming}

\begin{algorithm}[t]
\caption{Streaming Algorithm Phase I} \label{alg:streaming-phase-I}
\begin{algorithmic}[1]
\STATE \textbf{Input:} Precision $\varepsilon$, deletion parameter $d,$ sampling probability $p$ and parameter $\gamma$
\STATE $T \gets \{(1+\varepsilon)^{i} \mid i \in \mathbb{Z}\}$
\STATE $A \gets \emptyset$, $\Delta \gets 0$, $\tau_{min} \gets 0$, $B_{\tau} \gets \emptyset$ for all $\tau \in T$ 
\STATE $L\gets \emptyset, \Gamma \gets \emptyset, K \gets \emptyset, R\gets \emptyset, F\gets \emptyset$ \COMMENT{Auxiliary sets for the analysis}
\STATE $V_d \gets $ the first $d$ arriving elements
\FOR{every arriving element $e'$}
    \STATE Add $e'$ to $V_d$ and then pop element $e$ with smallest value from $V_d$
    \STATE $\Delta \gets \max\{f(e), \Delta\}$, $\tau_{min} \gets \frac{\varepsilon}{1+\varepsilon} \cdot \frac{\Delta}{k}, T \gets \{\tau \in T \mid \tau \ge \tau_{min}\}$
    \STATE Remove all $B_{\tau}$ s.t. $\tau \not \in T$\COMMENT{Deleted elements are added to $L$}
    \STATE\label{line:test} \textbf{if} $\tau_{min} > f(e\mid A)$ \textbf{then} Discard $e$ and \textbf{continue}\COMMENT{$L \gets L + e$\quad}
    \STATE Find the largest threshold $\tau \in T$ s.t. $f(e\mid A)\ge \tau$ and add $e$ to $B_{\tau}$
    \WHILE{$\exists \tau \in T$ such that $|B_{\tau}| \geq \frac{d}{\varepsilon}$
    }
        \STATE\label{line:sample}Remove one element $g$ from $B_\tau$ u.a.r. \COMMENT{$\Gamma \gets \Gamma + g$\quad}
        \STATE $w(g) \gets f(g\mid A)$
        \STATE Draw $X_{g}$ independently from a Bernoulli with parameter $p$
        \IF{$A + g \in \I$}
            \STATE \textbf{if} $X_g = 1$, \textbf{then} $A \gets A + g$ 
            \STATE \label{line:sampling1} \textbf{else} Discard $g$ \COMMENT{$R \gets R + g$\quad}
        \ELSE 
            \STATE $k_g \gets \argmin\{ w(k) \mid k \in C(A + g)\}$
            \IF{\label{line:swap} $w(g) > (1+\gamma) \cdot w(k_g)$}
                 \STATE \textbf{if} $X_g = 1$, \textbf{then} $A \gets A + g - k_g$  \COMMENT{$K \gets K + k_g$}
                \STATE \label{line:sampling2} \textbf{else} Discard $g$ \COMMENT{$R \gets R + g$\quad}
            \ELSE 
                \STATE Discard $g$ \COMMENT{$F \gets F + g$\quad }
             \ENDIF
        \ENDIF
        \STATE \textbf{Update} $\{B_{\tau}\}$ according to $A$ \COMMENT{Deleted elements are added to $L$}
    \ENDWHILE
\ENDFOR
\STATE $B \gets V_d \cup_{\tau \in T} B_{\tau}$
\RETURN $A, B$
\end{algorithmic}
\end{algorithm}

In this section we present our algorithm for Deletion Robust Submodular Maximization in the streaming setting. Here, the elements of $V$ in the first phase arrive on a stream and we want to compute a small summary $W$ using limited (online) memory. 
Our approach consists of carefully mimicking the swapping algorithm \citep{ChakrabartiK15} with subsampling \citep{FeldmanK018} in a deletion robust fashion; to that end, the algorithm maintains an independent candidate solution $A$ and buckets $B_{\tau}$ that contain {\em small} reservoirs of elements from the stream with similar marginal contribution each element with respect to the current solution $A.$ Beyond the $B_{\tau}$, an extra buffer $V_d$ containing the best $d$ elements seen so far is kept.

We start explaining the algorithm by defining the thresholds $\tau$. In the streaming setting we do not have an a priori estimate of $\OPT$, so that we do not know upfront which are the thresholds corresponding to {\em high quality} elements. This issue is overcome by initially considering {\em all} the powers of $(1+\eps)$ and progressively removing the ones too small with respect to $\Delta$, i.e., the $(d+1)$-st largest value seen so far.
%En passant, we mention that the role of the $B_{\tau}$ in the streaming and centralized algorithms have different flavors: while in the centralized case it is important to consider the $B_{\tau}$ in decreasing order of $\tau$ to mimic greedy, in the streaming setting we only need the $B_{\tau}$ to make sure that elements too small are not considered.  
Every new element $e'$ that arrives is first used to update $V_d$:  if its value is larger then the minimum in $V_d$ then we add it to $V_d$ and remove the smallest value element from $V_d$ and denote it by $e$. If the value of the newly arrived element is smaller than smallest value element in $V_d$, we keep the set $V_d$ intact. For simplicity, we set $e$ to be the newly arrived element $e'$ in this case. In other words, element $e$ is either the new element or the element which lost its place in $V_d$ to the new element.

In either case, we use element $e$ in the next steps of the algorithm. 
 Then, the value of $\Delta$ is updated if $f(e) > \Delta$  and all the buckets corresponding to thresholds that are too small are deleted, this guarantees that only a logarithmic number of active buckets is kept. At this point  element $e$ is put into the correct bucket $B_{\tau}$, if such bucket still exists. 
Now, new elements are drawn uniformly at random from the buckets $B_{\tau}$ with size at least $d /\eps$ as long as no bucket contains more than $d/\eps$ elements. These drawn elements are added to $A$ if and only if the subsampling is successful and it is either feasible to add them directly or they can be {\em swapped} with a {\em less important} element in $A$. To make this notion of importance more precise, each element in the solution is associated with a weight, i.e., its marginal value to $A$ when it was first considered to be added to the solution. To complete the notation, we let weight of a set of elements be the sum of the weights of its elements. 
An element in $A$ is swapped for a more promising one only if the new one has a weight that is larger by at least a $(1 + \gamma) $ factor and it is sampled independently with some probability $p$, while maintaining $A$ independent. 

Every time $A$ changes, the buckets $B_{\tau}$ are completely updated so to maintain the invariant that $B_{\tau}$ contains only elements whose marginal value with respect to the current solution $A$ is within $\tau$ and $(1+\eps) \cdot \tau.$
This property is crucial to ensure the deletion robustness: every bucket contains elements that are similar, i.e., whose marginal density with respect to the current solution is at most a multiplicative $(1+\eps)$ factor away. When the stream terminates, the algorithm passes the deletion robust summary $W$ (composed by the candidate solution $A$ and $B$, containing $V_d$ and the surviving buckets) to Algorithm~\ref{alg:phase-II}. The pseudocode of this algorithm is presented in \Cref{alg:streaming-phase-I}. Before stating and proving our Theorem, we describe some of the properties of the weight function $w$ that governs the swapping.

\begin{lemma}
\label{lem:killed_elements}
Let $w$ be the weight function, $A$ the candidate solution, $A' = A \setminus D$ and let $K$ be the set of all the elements that, at some point, were in $A$ but were later swapped. Then, the following properties hold:
\begin{itemize}
    \item[$(i)$] $\gamma \cdot w(K) \le w(A) \le f(A)$
    \item[$(ii)$] $w(A') \le f(A')$
    \item[$(iii)$] $f(A \cup K) \le w(A \cup K)$
\end{itemize}
\end{lemma}
\begin{proof}
    The proof of the first inequality is similar to the one of Lemma 9 in \cite{ChakrabartiK15}. Crucially, the weight function $w$ is linear and once an element enters $A$, its weight is fixed forever as the marginal value it contributed when entering $A$, moreover the weight of all the elements added to $A$ is strictly positive. During the run of the algorithm, every time an element $k_g$ is removed from $A$, the weight of $A$ increases by  $w(g) - w(k_g)$ by its replacement with some element $g$. Moreover, $\gamma \cdot w(k_g) \le w(g) - w(k_g)$ for every element $k_g \in K$ since $(1+\gamma)w(k_g) \leq w(g)$. Summing up over all elements in $K$, and recalling that the weight function is positive for all the elements in $A \cup K$, it holds that 
    \[
          w(K) = \sum_{k_g \in K}w(k_g) \leq \frac 1{\gamma}\sum_{k_g \in K} \left[w(g) - w(k_g)\right] \le \frac 1{\gamma} w(A).
    \]
    
    We show now the second inequality. Let $<a_1, a_2, \dots a_{\ell}>$ be the elements in $A$, sorted in the order in which they were added to $A$. We have that 
    \[
        f(A) = \sum_{i=1}^{\ell} f(a_i|\{a_1,\dots, a_{i-1}\}) \ge \sum_{i=1}^{\ell} f(a_i|A_{a_i}) = \sum_{i=1}^{\ell} w(a_i) = w(A),
    \]
    where $A_{a_i}$ is the solution set right before $a_i$ is added to $A$. The inequality follows from submodularity, since $\{a_1,\dots, a_{i-1}\} \subseteq A_{a_i}$.
    Similarly, let $I(a)$ the indicator random variable corresponding to $a \not \in D$, while $I(a)\cdot a$ is a shorthand for the element $a$ if $I(a)=1$ and the empty set otherwise.
    \begin{align*}
        f(A') &= \sum_{i=1}^{\ell} I(a_i)\cdot f(a_i|\{I(a_1)\cdot a_1,\dots, I(a_{i-1})\cdot a_{i-1}\}) \\
        &\ge \sum_{i=1}^{\ell}  I(a_i)\cdot f(a_i|\{a_1,\dots,a_{i-1}\}) \\
        &\ge \sum_{i=1}^{\ell}  I(a_i)\cdot f(a_i|A_{a_i}) = w(A').
    \end{align*}
    For the last inequality, let $x_1, x_2, \dots x_{t}$ be the elements in $A \cup K$, sorted by the order in which they were added to $A$. We have that 
    \[
        f(A \cup K) = \sum_{i=1}^{t} f(x_i|\{x_1,\dots, x_{i-1}\}) \le \sum_{i=1}^{t} f(x_i|A_{x_i}) = \sum_{i=1}^{t} w(x_i) = w(A \cup K),
    \]
    where $A_{x_i}$ is the solution set right before $x_i$ is added to $A$. The inequality follows from submodularity: $A_{x_i}$ contains all the elements entered in $A$ before $x_i$ minus those elements that have already been removed from it.
\end{proof}

We are now ready for the main result of the Section. Instead of stating a general black box result as in \Cref{thm:robust-centralized}, we consider directly the algorithm from \cite{BuchbinderF19} as \BBAlg and analyze the relative algorithm for a suitable choice of $p$ and $\gamma.$ This is just for clarity of exposition: the dependence on $\beta$ of the generic approximation is more involved than in the centralized case. 

\begin{theorem}
\label{thm:robust-streaming}
    Consider the problem of deletion robust non-monotone submodular maximization with matroid constraints in the streaming scenario. Fix any constant $\delta \in (0,1)$, then there exists a constant $G_{\delta}$ such that for any $\eps \in (0,\delta)$, in expectation a $(9.435 + G_\delta \cdot \eps)$-approximation algorithm with summary and memory size $O(k + \frac{d}{\eps^2}\log \frac{k}{\eps})$ exists.
\end{theorem}
\begin{proof}
    We start by bounding the memory of the algorithm. Sets $A$ and $V_d$ always contains at most $k$, respectively $d$, elements. Every time a new element of the stream is considered, all the active $B_{\tau}$ contain at most $d/\eps$ elements; furthermore there are always at most $O(\frac 1{\eps}\log \frac{k}{\eps})$ of them. This is ensured by the invariant that the active thresholds are smaller than $\Delta$ (by submodularity) and larger than $\tau_{\min}.$ Overall, the memory of the algorithm and the summary size is thus $O(k + \frac{d}{\eps^2}\log \frac{k}{\eps})$.
    
     As in the analysis of \Cref{thm:robust-centralized}, we now fix any set $D$ and study the relative expected performance of our algorithm. In the pseudocode we have introduced some auxiliary sets that are useful in the analysis: the base set $V$ is partitioned into the elements that are considered at some point in line \ref{line:sample} (contained in $\Gamma$), those in the buckets or $V_d$ (contained in set $B$) and the ones with low marginals (contained in set $L$). 
    % To do so we need some notation: let $K$ be the set of all elements removed from the solution $A$ at some point; moreover, for any element $g$ considered at some point in line \ref{line:sample} of \Cref{alg:streaming-phase-I}, let $A_g$ denote the candidate solution in that moment (possibly containing the element $k_g$ that was swapped with $g$). Given $A, B'$ and $K$, slightly differently than in the centralized case we define the set of good elements $C$ as the elements that passed the test in line \ref{line:test}: $ C = \left\{e \in V \mid e \text{ passed test in line \ref{line:test}}\right\}.$
    We have all the ingredients to decompose $\OPT \cup A \cup K$ accordingly:
    \begin{align}
    \nonumber
        (1-p)f(\Sopt) \le& \E{f(\Sopt \cup A \cup K)} \\
    \nonumber
    =& \E{f(A \cup K)} + \E{f(\Sopt \cap L\mid A \cup K)} \\
    \label{eq:str-decomposition}
        &+ \E{f(\Sopt \cap B \mid A \cup K)} +  \E{f(\Sopt \cap \Gamma \mid A \cup K)}
    \end{align}
    Note that the first inequality follows from the uniform sampling step and the sampling Lemma: each fixed element belong to $A \cup K$ with a probability that is at most $p$. 
    
    The second and third terms of the decomposition are fairly easy to bound. If an element $e$ is added to $L$ at some point it means that its marginal contribution with respect to the current solution $A$ is smaller than the current estimate $\tau_{min}$ of $\eps f(\OPT)/k$ (of which it is consistently a lower bound), by submodularity this implies that also its contribution with respect to the final $A \cup K$ is such that $f(e|A \cup K) \le \eps f(\OPT)/k$, all in all:
    \begin{equation}
    \label{eq:str-small-els}
        f(\OPT \cap L \mid A \cup K) \le \sum_{x \in \OPT \cap L} f(x| A \cup K) \le \eps \cdot f(\OPT).
    \end{equation}
    
    We know that $\OPT \cap D = \emptyset$, this implies that $\OPT \cap B$ is contained in $B'$ and thus in the summary $W$ that is passed to the second phase algorithm. Plus, note that $\OPT \cap B$ is independent because it is contained in $\OPT$. So, if we use as $\BBAlg$ a routine for centralized submodular maximization with matroid constraint that guarantees a $\beta$ approximation, we have:
    \begin{align}
    \label{eq:str-OPTcapB}
        \E{f(\Sopt \cap B\mid A \cup K)} \le \E{f(\Sopt \cap B)} \le \beta \cdot \E{f(S)}. 
    \end{align}
    
    At this point we get back to the decomposition in \Cref{eq:str-decomposition} and bound the first term. We show that the value of all the elements that {\em at some point} were in the candidate solutions, i.e.,$(A \cup K)$, is comparable to (a constant factor of) $f(A')$. This result handles two challenges: it shows that the total value of the elements swapped is upper bounded by $1/\gamma$ that of those that are not swapped, and that $A$ is indeed robust to deletion. 

\begin{lemma}
\label{lem:streaming_robust}
    It holds that $\E{f(A\cup K)} \le \frac{1+\eps}{1-\eps} (1+\frac 1{\gamma}) \E{f(A')}$.
\end{lemma}
\begin{proof}
    As a first step, we show that $w(A)$ is close to $f(A')$, at least in expectation. Let $A_{\tau}$ be the subset of elements that were added into $A$ coming from $B_{\tau}$ and $A'_{\tau} = A_{\tau} \setminus D$. Moreover, let $a_t^{\tau}$ be the $t^{th}$ element added to $A_{\tau}$ (if any). We have the following:
    \begin{align*}
        \E{|A'_{\tau}|} &= \E{\sum_{t=1}^{|A_{\tau}|} \mathbbm{1}(a_t^{\tau} \not \in D) } = \sum_{t=1}^{+\infty}\E{ \mathbbm{1}(a_t^{\tau} \not \in D) \cdot \mathbbm{1}(|A_{\tau}| \ge t) } \\
        &=\sum_{t=1}^{+\infty} \E{\mathbbm{1}(a_t^{\tau} \not \in D)\big||A_{\tau}| \ge t} \cdot \P{|A_{\tau}| \ge t} \\
        &\ge (1-\varepsilon)\sum_{t=1}^{+\infty}\P{|A_{\tau}| \ge t}\\
        &=(1-\varepsilon)\E{|A_{\tau}|}.
    \end{align*}
    The crucial observation is that when the algorithm decides to add an element to $A$ from $B_{\tau}$, then the probability that the element belongs to $D$ is at most $\eps.$ Let $A_a$ be the elements in $A$ when $a$ is added, so that $w(a) = f(a|A_a)$. Moreover $I(a)$ is the indicator variable of the event $a \in A'$ given that $a \in A$. We have
    \[
        w(A) = \sum_{\tau \in T} \sum_{a \in A_{\tau}} w(a)\text{, while } w(A') = \sum_{\tau \in T} \sum_{a \in A_{\tau}} I(a) \cdot w(a) 
    \]
    We know that the weight of an element $a$ coming from $B_t$ is such that $\tau \le w(a) \le (1+\eps)\cdot \tau$. Therefore we can proceed similarly to what we had in \Cref{lem:robustness}:
    \begin{align}
    \nonumber 
        \E{w(A')} &= \E{\sum_{\tau \in T} \sum_{a \in A_{\tau}} I(a) \cdot w(a)} = \sum_{\tau \in T} \E{\sum_{a \in A_{\tau}} I(a) \cdot w(a)} \ge \sum_{\tau \in T} \tau \cdot \E{\sum_{a \in A_{\tau}} I(a)}\\
    \label{eq:robustness_streaming}
        &= \sum_{\tau \in T} \tau \cdot \E{|A'_{\tau}|} \ge  (1-\varepsilon) \cdot \sum_{\tau \in T} \tau\cdot \E{|A_{\tau}|} \ge \frac{(1-\varepsilon)}{(1+\varepsilon)} \E{w(A)}.
    \end{align}
    Now, let's move our attention to $w(A) + w(K)$.
    Using the properties as in \Cref{lem:killed_elements} of the weight function we have that
    \begin{align}
    \nonumber
        \E{f(A \cup K)} &\le \E{w(A \cup K)}\tag*{(Property (iii))}
        \\ 
        &= \E{w(A) + w(K)}
        \tag*{(Linearity of $w$)}        \\ 
        &\le \left( 1 + \frac 1{\gamma}\right) \E{w(A)}
        \tag*{(Property (i))}\\ 
        &\le \left( 1 + \frac 1{\gamma}\right)  \frac{1+\eps}{1-\eps} \cdot \E{w(A')}
        \tag*{(Due to (\ref{eq:robustness_streaming}))}\\ 
        &\le \left( 1 + \frac 1{\gamma}\right)  \frac{1+\eps}{1-\eps} \cdot \E{f(A')}
        \tag*{(Property (ii))}
        \\
        \label{eq:str-AcupK}
        &\le \left( 1 + \frac 1{\gamma}\right)\frac{1+\eps}{1-\eps}\cdot \E{f(S)}
    \end{align}
    \end{proof}
    
    We are left with bounding the marginal contribution of elements in $\OPT \cap \Gamma$ with respect to $A \cup K$. $\Gamma \setminus(A \cup K)$ is partitioned into two disjoint sets: $F$ that is the set of all the elements that failed the swapping test (line \ref{line:swap}) and $R$ that is the set of all the elements that were removed because of sampling either in line \ref{line:sampling1} or \ref{line:sampling2}.
    
\begin{lemma}
\label{eq:str-swapping}
Choosing $p = 1/(\gamma + 2)$, it holds that  
\[
        \E{f(\Sopt \cap \Gamma \mid A \cup K)} \le (1+\gamma) \frac{1+\eps}{1-\eps}\cdot \E{f(S)}.
\]
\end{lemma}
\begin{proof}
    Let $g$ be any element in $\OPT \cap \Gamma$, and fix any compatible story $\cE_g$ of the algorithm up to the point when $g$ is considered in line \ref{line:sample} of \Cref{alg:streaming-phase-I}. The story $\cE_g$ deterministically induces the candidate solution $A_g$ when $g$ is considered. If $g$ is not in $F$, i.e. it is possible to either add it directly or swap it with some $k_g \in A_g$ such that $f(g \mid A_g) = w(g) \ge (1+\gamma) w(k_g)$, then we have two cases: either $g$ is successfully sampled and swapped with some element in the current candidate solution $A_g$, or $g$ is added to $R$ and discarded. Conditioning on $\cE_g$, we have:
    \[
        p \cdot \E{\ind{g \in R} w(g) \mid \cE_g} = p \cdot (1-p) \E{w(g) \mid \cE_g} = (1-p) \cdot \E{\ind{g \in A \cup K}  w(g) \mid \cE_g}.
    \]
    Summing over all such $g$ and taking the expectation over all the compatible $\cE_g$, we get:
    \begin{equation}
    \label{eq:sampled}
        p \cdot \E{w(\OPT \cap R)}  = (1-p) \cdot \E{w(\OPT \cap (A \cup K))}.
    \end{equation}
    We now move our attention to the elements in $\OPT \cap F$, i.e., the good elements in $\OPT$ that have been discarded because of the matroid constraint. As a first step we show the existence of the injection $h:(\OPT \cap \Gamma) \setminus R \to A$ with the following properties:
    \begin{itemize}
        \item[$(a)$] $h$ restricted to $A \cap \OPT$ is the identity
        \item[$(b)$] $w(g) \le (1+\gamma) w(h(g))$ for all $g \in \OPT \cap F$ 
        \item[$(c)$] $w(g) \le w(h(g))$ for all $g \in (\OPT \cap K)$ 
    \end{itemize}
    
    We prove the existence of such $h$ by constructing a suitable graph $G$ where we then apply \Cref{lem:graph}. The nodes of $G$ are the elements of $\Gamma$, while the edges are created iteratively as the algorithm considers elements in $A$. Let $u$ be a generic element arriving in $\Gamma \setminus R$, and let $A_u$ be the candidate solution at that point. If $A_u + u \in \cM$, then no edge is created at that iteration, otherwise, there exists a unique cycle $C(A_u+u)$ in $A_u + u$ containing $u$. Now, we have two cases, either $u$ is swapped with some element $x$ in $C(A_u + u) \cap A_u$, in which case we create directed edges from $x$ to each element in $C(A_u + u) - x$, or $u$ is rejected. In the latter case, we create directed edges from $u$ to each element in $C(A_u + u) - u$. Note that, in both cases, the out-neighborhood of the node that gets discarded spans the element corresponding to the node itself.
    It is also clear that the graph $G$ generated at the end of the procedure is acyclic, as an out-edge is only created when an element is discarded towards elements that have not been discarded yet. Thus, the graph respects the assumption of \Cref{lem:graph} and there exists an injective function $h$ such that, every vertex/element $u \in (\OPT \cap \Gamma)\setminus R$ is associated with a sink $h(u)$ that is reachable from $u$. From a different perspective, this $h$ is a charging function that associates each element in $(\OPT \cap \Gamma)\setminus R$ with an element in the final solution $A$ that \emph{accounts} for it. It is easy to see that $h$ respects the three properties we want to enforce. $h$ restricted to $A \cap \OPT$ is the identity as elements in $A$ are sinks. Consider now any element $g \in \OPT \cap F$, this is an element that has been discarded without being added to the solution, this means that outgoing edges from $g$ are such that $w(g) \le (1+\gamma) w(u)$, for all $u \in \delta^+(g).$ If we focus on the unique path from $g$ to $h(g)$, we see that after the first step, when the weight decreases by at most a $(1+\gamma)$ factor, all the following edges do not decrease the weight, this is because an outgoing edge from an element of the solution either points to an element it is swapped with (whose weight is larger by at least a $(1+\gamma)$ factor) or to another element in the solution but with larger weight (as we swap the element with smallest weight in the cycle). The same argument clearly implies that, for all the elements already in the solution that are later swapped, i.e. in $\OPT \cap K$, it holds that the weight does not decrease along the path to the associated sink. We can use the properties of the injective function $h$ to bound $w(\OPT \cap F)$.
    \begin{align}
        \mathbb{E}&[w(\OPT \cap F)] \le (1 + \gamma) \E{w(h(\OPT \cap F))} \tag*{Property $(b)$} \\
        &\le (1 + \gamma) \E{w(h(\OPT \cap F)) + w(h(\OPT \cap K))- w(\OPT \cap K))} \tag*{Property $(c)$} \\
        \nonumber
        &=(1 + \gamma) \E{w(h(\OPT \cap (F \cup K)) - w(\OPT \cap K) + w(\OPT \cap A) - w(\OPT \cap A)}\\
        &=(1 + \gamma) \E{w(h((\OPT \cap \Gamma)\setminus R)) - w(\OPT \cap (A \cup K)))}\tag*{Property $(a)$}\\
        \label{eq:explaination}
        &\le (1 + \gamma) \E{w(A)} - (1+\gamma)(w(\OPT\cap (A \cup K)))
    \end{align}
    Observe that in the last inequality we use that $w(a) \ge 0$ for all $a \in A$, as $h((\OPT \cap \Gamma)\setminus R)$ may be a strict subset of $A.$ We can finally proceed to prove the Lemma.

    \begin{align*}
        \mathbb{E}[f(\Sopt \cap \Gamma&\mid A \cup K)] \le \E{\sum_{x \in \Sopt \cap F} f(x | A \cup K)} + \E{\sum_{x \in \Sopt \cap R} f(x | A \cup K)}\\
        \le & \E{w(\Sopt \cap F)} + \E{w(\Sopt \cap R)} \\
        \le & (1+\gamma) \E{w(A)} + \left(\frac{1-p}{p} - 1 - \gamma\right)\E{w(\OPT \cap (A \cup K))}\tag*{By (\ref{eq:sampled}) and (\ref{eq:explaination})}\\
        = & (1+\gamma)\E{w(A)}\tag*{By the choice of $p$}\\
        \le & (1+\gamma)\frac{1+\eps}{1-\eps}\cdot \E{f(S)}\tag*{As in \Cref{lem:streaming_robust}}.
    \end{align*}
    Note that the first two inequalities follow from submodularity and the fact that the candidate solution when an element is considered (and thus that is used to compute its weight) is always contained in $A \cup K$. 
\end{proof}
    Plugging \Cref{eq:str-small-els,eq:str-AcupK,eq:str-OPTcapB,eq:str-swapping} into \Cref{eq:str-decomposition} we obtain the following:
    \begin{align*}
        \left(1-\frac{1}{2+\gamma}\right)f(\Sopt) \le& \E{f(A \cup K)}
        +\E{f(\Sopt \cap L|A \cup K)}\\
        &+ \E{f(\Sopt \cap B|A \cup K)}+ \E{f(\Sopt \cap \Gamma|A \cup K)}\\
        \le& \left[\left(2+\gamma +\frac1 \gamma\right)\frac{1+\eps}{1-\eps} + \beta \right]\cdot\E{f(S)} + \eps \cdot f(\OPT)
    \end{align*}
    Rearranging terms, one gets
    \begin{align*}
    % \label{eq:calculations_streaming}
        f(\OPT) &\le \left(1 - \eps - \frac{1}{2+\gamma}\right)^{-1}\left[\left(2+\gamma +\frac1 \gamma\right)\frac{1+\eps}{1-\eps} + \beta \right]\cdot \E{f(S)}\\
        &= \frac{(2+\gamma)[\gamma^2 + (\beta + 2 ) \gamma + 1 + \eps(\gamma^2 + (2 - \beta)  \gamma + 1)]}{\gamma (1-\gamma)(1-\eps)[1+\gamma -\eps (\gamma+2)]}\cdot\E{f(S)}\\
        &= \frac{(2+\gamma)[\gamma^2 + (\beta + 2 ) \gamma + 1]}{\gamma[1+\gamma -\eps (\gamma+2)]}\frac{1+\eps}{1-\eps}\cdot\E{f(S)}
    \end{align*}
    The Theorem follows by using as optimization routine the state of the art algorithm for centralized (non-monotone) submodular maximization with matroid constraint ($\beta \approx 2.597$ as in \cite{BuchbinderF19}) and then choose $\gamma$ (and accordingly $p$) as to minimize the multiplicative factor. In particular, if we set $\gamma \approx 1.746$, then for any fixed $\delta$ there exists a constant $G_{\delta}$ such that we get an approximation of the form $(9.435 + G_{\delta} \cdot \eps)$ for all $\eps$ in $(0,\delta).$
\end{proof}

    As a corollary of the previous Theorem, we get an improved approximation considering monotone objective. There is clearly no need for the sampling procedure and we can use the greedy \cite{fisher78-II} or continuous greedy \citet{CalinescuCPV11} subroutine  as $\BBAlg$. 
\begin{corollary}
\label{cor:streaming}
    Consider the problem of deletion robust monotone submodular maximization with matroid constraints in the streaming setting and fix any constant $\delta \in (0,1)$, then there exists a constant $\tilde G_{\delta}$ such that for any $\eps \in (0,\delta)$, in expectation a ${(5.582+\eps \cdot \tilde G_{\delta})}$-approximation algorithm with summary size and memory $O(k + \frac{d}{\eps^2}\log \frac{k}{\eps})$ exists.
\end{corollary}
\begin{proof}
    The proof is similar to the one of \Cref{thm:robust-streaming}. The only difference is that we do not do subsampling as $f(\OPT)$ is smaller than $f(\OPT \cup A \cup K)$ by monotonicity. Thus, we set $p=1$ and $\gamma = 1$ and then \Cref{eq:str-decomposition} becomes:
    \begin{align*}
        f(\Sopt) \le& \E{f(\Sopt \cup A \cup K)} \\
    =& \E{f(A \cup K)} + \E{f(\Sopt \cap L\mid A \cup K)} \\
        &+ \E{f(\Sopt \cap B \mid A \cup K)} +  \E{f(\Sopt \cap \Gamma \mid A \cup K)} \\
        \le& 2\frac{1+\eps}{1-\eps} \E{f(S)} + \eps f(\OPT)\\
        &+ \beta \cdot \E{f(S)} + \E{f(\OPT \cap \Gamma|A \cup K)}
    \end{align*}
    For the last term, we use a more general statement. For all $X \subseteq \Gamma,$  $X\in \cM$, then it holds that $
        w(X) \le 2 \cdot w(A).$
    This is a direct application of Theorem 1 of \citet{Varadaraja11} (using a single matroid and setting $r = 2$). More in specific, we can imagine to restrict the stream to consider only the elements in $\Gamma$, with the order in which they are considered in line \ref{line:sample}.
     Choosing $X = \OPT \cap \Gamma$, we then get that 
        \[
            \E{f(\OPT \cap \Gamma|A \cup K)} \le \E{w(\OPT \cap \Gamma)} \le 2 \E{w(A)}.
        \]
    We use as optimization routine \BBAlg continuous greedy, therefore we can plug in $\beta = \frac{e}{e-1}$ and, rearranging the terms we obtain
    \begin{align*}
        f(\Sopt) &\le \left(2\frac{1+\eps}{(1-\eps)^2} + \frac{3e  -2}{(1-\eps)(e-1)}\right) \E{f(S)}\\
        &\le \Big{[}\frac{e}{e-1} + 4  + \underbrace{\frac{9 e - 8 -\delta(5e-4)}{(e-1) (1-\delta)^2}}_{\tilde{G}_{\delta}} \cdot \eps \Big{]}\E{f(S)}.
    \end{align*}
    The last inequality can be numerically verified and holds for any  $\eps \in (0,\delta).$
    
\end{proof}

\section{Conclusion and Future Work}

We presented the first space-efficient constant-factor approximation algorithms for deletion robust submodular maximization over matroids in both the centralized and the streaming setting. In particular, we are also the first to design space-efficient deletion robust algorithms for non-monotone objective with {\em any} type of constraints.
A natural direction for future work is to extend these results and ideas to other constraints (e.g., multiple matroids and knapsack), and to consider fully dynamic versions with insertions and deletions.

\section*{Acknowledgements}
Federico Fusco’s work was partially supported by the ERC Advanced Grant 788893 AMDROMA
“Algorithmic and Mechanism Design Research in Online Markets” and the MIUR PRIN project
ALGADIMAR “Algorithms, Games, and Digital Markets.” Part
of this work was done while Federico was an intern at Google
Research, hosted by Paul D\"utting.

\bibliographystyle{plainnat}
\bibliography{cite}

\end{document}